\documentclass[a4paper,fleqn,11pt]{article}

\usepackage{amsmath}
\usepackage{amsthm}
\usepackage{amssymb}
\usepackage{accents}

\usepackage[a4paper,top=3cm, bottom=3cm, left=3cm, right=3cm]{geometry}
\usepackage[shortlabels,inline]{enumitem}
\usepackage[pdftex,ocgcolorlinks,pagebackref=false]{hyperref}
\usepackage[affil-it]{authblk}

\setlist[enumerate,1]{label={(\roman*)}}

\theoremstyle{plain}
\newtheorem{theorem}{Theorem}[section]
\newtheorem{lemma}[theorem]{Lemma}
\newtheorem{proposition}[theorem]{Proposition}
\newtheorem{corollary}[theorem]{Corollary}

\theoremstyle{definition}

\newtheorem{definition}[theorem]{Definition}

\newcommand{\ed}{\mathop{}\!\mathrm{d}}

\newcommand{\norm}[2][]{\left\|#2\right\|_{#1}}
\newcommand{\ket}[1]{\left|#1\right\rangle}

\newcommand{\ketbra}[2]{\left|#1\middle\rangle\!\middle\langle#2\right|}

\newcommand{\setbuild}[2]{\left\{#1\middle|#2\right\}}
\DeclareMathOperator{\distributions}{\mathcal{P}}
\DeclareMathOperator{\boundeds}{\mathcal{B}}
\newcommand{\positivedefiniteoperators}[1]{\boundeds(#1)_{++}}
\DeclareMathOperator{\states}{\mathcal{S}}

\DeclareMathOperator{\Tr}{Tr}
\DeclareMathOperator{\Hom}{Hom}

\DeclareMathOperator{\spectrum}{spec}
\DeclareMathOperator{\id}{id}

\newcommand{\maxrelativeentropy}[3][]{D^{#1}_{\textnormal{max}}\mathopen{}\left(#2\middle\|#3\right)\mathclose{}}
\newcommand{\maxdistance}{d_{\textnormal{T}}}

\newcommand{\reals}{\mathbb{R}}
\newcommand{\complexes}{\mathbb{C}}
\newcommand{\naturals}{\mathbb{N}}
\newcommand{\nonnegativereals}{\mathbb{R}_{\ge 0}}
\newcommand{\positivereals}{\mathbb{R}_{>0}}
\newcommand{\positiveintegers}{\mathbb{N}_{>0}}

\newcommand{\preorderle}{\preccurlyeq}
\newcommand{\preorderge}{\succcurlyeq}

\newcommand{\asymptoticge}{\succsim}

\newcommand{\catalyticge}{\preorderge_{\textnormal{c}}}

\newcommand{\familysemiring}[2]{S_{#1,#2}}
\newcommand{\cfamilysemiring}[2]{S^{\textnormal{c}}_{#1,#2}}
\newcommand{\testspectrum}{\mathrm{TSper_1}}

\newcommand{\tropicals}{\mathbb{TR}}

\newcommand{\thesandwicheddivergence}{\widetilde{D}}
\newcommand{\sandwicheddivergence}[3]{\thesandwicheddivergence_{#1}\left(#2\middle\|#3\right)}
\newcommand{\sandwichedquasientropy}[3]{\widetilde{Q}_{#1}\left(#2\middle\|#3\right)}
\newcommand{\sandwichedmeandivergence}[3]{\thesandwicheddivergence_{#1}\left(#2\middle\|#3\right)}

\DeclareMathOperator{\gmeans}{\mathcal{G}}
\DeclareMathOperator{\Bohrcpt}{b}
\DeclareMathOperator{\poly}{\textnormal{poly}}

\title{Equivariant relative submajorization}
\author[1,2]{Gergely Bunth}
\author[1,2]{P\'eter Vrana}
\affil[1]{Institute of Mathematics, Budapest University of Technology and Economics, Egry~J\'ozsef u.~1., Budapest, 1111 Hungary.}
\affil[2]{MTA-BME Lend\"ulet Quantum Information Theory Research Group}

\begin{document}
\maketitle
\begin{abstract}
We study a generalization of relative submajorization that compares pairs of positive operators on representation spaces of some fixed group. A pair equivariantly relatively submajorizes another if there is an equivariant subnormalized channel that takes the components of the first pair to a pair satisfying similar positivity constraints as in the definition of relative submajorization. In the context of the resource theory approach to thermodynamics, this generalization allows one to study transformations by Gibbs-preserving maps that are in addition time-translation symmetric. We find a sufficient condition for the existence of catalytic transformations and a characterization of an asymptotic relaxation of the relation. For classical and certain quantum pairs the characterization is in terms of explicit monotone quantities related to the sandwiched quantum R\'enyi divergences. In the general quantum case the relevant quantities are given only implicitly. Nevertheless, we find a large collection of monotones that provide necessary conditions for asymptotic or catalytic transformations. When applied to time-translation symmetric maps, these give rise to second laws that constrain state transformations allowed by thermal operations even in the presence of catalysts.
\end{abstract}

\section{Introduction}

A pair of positive vectors $(p,q)\in\nonnegativereals^d\times\nonnegativereals^d$ is said to relatively submajorize another pair $(p',q')$ if there exists a substochastic map $T$ such that $T(p)\ge p'$ and $T(q)\le q'$ componentwise \cite{renes2016relative}. This relation can be used to characterize probabilistic and work-assisted thermal operations between incoherent states, as well as error probabilities in hypothesis testing. However, conditions based on relative (sub)majorization (or thermo-majorization) are insufficient to characterize thermal transformations in the presence of quantum coherence \cite{lostaglio2015description}.

Quantum majorization is a relation between bipartite quantum states sharing a marginal. A state $\rho_{AB}$ quantum majorizes $\rho'_{AB'} $ if there is a quantum channel $T:\boundeds(\mathcal{H}_B)\to\boundeds(\mathcal{H}_{B'})$ such that $\rho'_{AB'}=(\id_A\otimes T)(\rho_{AB})$. In \cite{gour2018quantum} it was shown that this relation, as well as a $G$-covariant version (for some compact group $G$) can be characterized using an infinite family of monotones defined in terms of the conditional min-entropy. For specific classical-quantum states (with $A$ a classical bit), quantum majorization with covariance encodes time-translation symmetric Gibbs-preserving transformations which, like thermal operations, puts constraints on the evolution of states with coherence between energy eigenstates.

In this paper, focusing on the classical-quantum case, we study transformations between pairs of positive operators by equivariant maps in a sense similar to relative submajorization: given representations $\pi:G\to U(\mathcal{H})$ and $\pi':G\to U(\mathcal{H}')$, and pairs of positive operators $(\rho,\sigma)$ on $\mathcal{H}$ and $(\rho',\sigma')$ on $\mathcal{H}'$, we say that $(\pi,\rho,\sigma)$ equivariantly relatively submajorizes $(\pi',\rho',\sigma')$ if there is a completely positive trace-nonincreasing map $T$ that is equivariant, ie. satisfies $T(\pi(g)A\pi(g)^*)=\pi'(g)T(A)\pi'(g)^*$ for all $g\in G$ and operator $A$, in addition to the inequalities $T(\rho)\ge\rho'$ and $T(\sigma)\le\sigma'$.

An averaging argument shows that this relation can equivalently be understood as transformations between families of positive operators parametrized by two copies of $G$. Somewhat more generally, we will consider pairs of continuous families of positive operators $\rho:X\to\positivedefiniteoperators{\mathcal{H}}$, $\sigma:Y\to\positivedefiniteoperators{\mathcal{H}}$, where $X$ and $Y$ are fixed nonempty compact topological spaces (when studying $G$-equivariant transformations for a compact group $G$, one would use $X=Y=G$). In this case we say that $(\sigma,\rho)$ relatively submajorizes $(\sigma',\rho')$ (notation: $(\sigma,\rho)\preorderge(\sigma',\rho')$) if there is a completely positive trace-nonincreasing map $T$ such that $T(\rho(x))\ge\rho'(x)$ and $T(\sigma(y))\le\sigma'(y)$ for all $x\in X$ and $y\in Y$.

Our main result is a characterization of an asymptotic relaxation of this relation and a sufficient condition for the possiblity of a catalytic transformation. We say that $(\rho,\sigma)$ asymptotically relatively submajorizes $(\rho',\sigma')$ if $(2^{o(n)}\rho^{\otimes n},\sigma^{\otimes n})\preorderge({\rho'}^{\otimes n},{\sigma'}^{\otimes n})$. Assuming that the image of $\sigma$ and $\sigma'$ consist of commuting operators, the characterization is in terms of explicitly given monotones: $(\rho,\sigma)$ asymptotically relatively submajorizes $(\rho',\sigma')$ iff the inequalities
\begin{equation}\label{eq:commutingasymptoticconditions}
    \sandwicheddivergence{\alpha}{\rho(x)}{\exp\int_Y\ln\sigma\ed\gamma}\ge\sandwicheddivergence{\alpha}{\rho'(x)}{\exp\int_Y\ln\sigma'\ed\gamma}
\end{equation}
hold for every $\alpha\ge 1$, $x\in X$ and probability measure $\gamma$ on $Y$, where $\thesandwicheddivergence_\alpha$ is the minimal (or sandwiched) R\'enyi divergence \cite{muller2013quantum,wilde2014strong}. If the inequalities are strict and $\Tr\rho(x)>\Tr\rho'(x)$ for all $X$, then relative submajorization holds after tensoring both pairs with a suitable catalyst. Without the commutativity assumption, we find generalizations of the conditions \eqref{eq:commutingasymptoticconditions} that are necessary for asymptotic or catalytic ordering. In these the second argument is replaced with a suitable non-commutative geometric mean. For example,
\begin{equation}\label{eq:distancefromgeometricmean}
    \sandwicheddivergence{\alpha}{\rho(x)}{\sigma(y_1)\#\sigma(y_2)}
\end{equation}
is one of these monotones, where $x\in X$, $y_1,y_2\in Y$, and $\sigma(y_1)\#\sigma(y_2)$ is the matrix geometric mean \cite{pusz1975functional}.

To prove our results, we use recent results from the theory of preordered semirings to find conditions in terms of monotone quantities that are additive under direct sums and multiplicative under the tensor product, following some of the ideas of \cite{perry2020semiring,bunth2021asymptotic}. While in general these monotones are defined only implicitly, under the additional assumption that the image of $\sigma$ consists of commuting operators, we obtain a complete classification, identifying them as exponentiated sandwiched R\'enyi divergences between one of the $\rho$ operators and a weighted geometric mean of the $\sigma$ operators. Finding all the relevant quantum extensions appears to be a difficult problem, although our heuristic approach reveals a way to systematically construct some of them. Interestingly, these also give new monotones for pair transformations by specialization: for example, it follows that
\begin{equation}
    (\rho,\sigma)\mapsto\sandwicheddivergence{\alpha}{\rho}{\rho\#\sigma}
\end{equation}
is a quantity that satisfies the data processing inequality (although it is not a monotone under relative submajorization).

The paper is structured as follows. In Section~\ref{sec:preliminaries} we collect some results on integration of continuous functions on compact spaces and on preordered semirings. In Section~\ref{sec:relsubfamilies} we study relative submajorization of pairs of families: in Section~\ref{sec:familysemiring} we define the precise setup and show that the preorder on the semiring of pairs of families satisfies the required technical conditions; in Section~\ref{sec:classical} we provide a classification of the relevant monotones when restricted to classical families; in Section~\ref{sec:quantum} we extend the classification to arbitrary $\rho$ and commuting $\sigma$, and construct some monotones for general pairs. Section~\ref{sec:applications} discusses some applications of our results: in Section~\ref{sec:equivariant} we describe in more detail how equivariant relative submajorization (including non-compact groups) can be encoded as the relative submajorization of certain families, with applications to time-translation symmetric Gibbs-preserving maps and to group-symmetric hypothesis testing; in Section~\ref{sec:approximate} we relate a type of approximate asymptotic transformation to asymptotic relative submajorization; in Section~\ref{sec:divergences}, using the monotones in the fully quantum case, we find a new family of monotone quantum R\'enyi divergences.

\section{Preliminaries}\label{sec:preliminaries}

When $\mathcal{H}$ is a Hilbert space, $\positivedefiniteoperators{\mathcal{H}}$ denotes the set of positive definite operators on $\mathcal{H}$. Our convention is that this includes the zero operator on zero dimensional Hilbert spaces.

We will make use of some facts on positive functionals on $C(X)$ for compact Hausdorff $X$ (see e.g. \cite[Chapter 7]{folland1999real}). On such a space, a Radon measure is a finite regular Borel measure.

\begin{theorem}[Riesz representation theorem, {\cite[7.2 Theorem]{folland1999real}}]\label{thm:Riesz}
Let $X$ be a compact Hausdorff topological space and $L:C(X)\to\reals$ a positive linear functional. Then there exists a unique Radon measure $\mu$ on $X$ such that for all $\xi\in C(X)$ the equality
\begin{equation}\label{eq:measurerepresentation}
L(\xi)=\int_X\xi(x)\ed\mu(x)
\end{equation}
holds. Conversely, every Radon measure gives rise to a positive linear functional via \eqref{eq:measurerepresentation}.
\end{theorem}
Examples of Radon measures include positive linear combinations of Dirac measures and Haar measures of locally compact topological groups. On a compact space $X$, every Radon measure $\mu$ is in the closure (with respect to the vague topology) of the set of positive linear combinations of Dirac measures with total mass $\mu(X)$ \cite[III, \S 2, No. 4, Cor. 3.]{bourbaki2004integrationi}.

A \emph{preordered semiring} $(S,+,\cdot,0,1,\preorderle)$ consists of a set $S$, two commutative and associative binary operations $+,\cdot:S\times S\to S$ that satisfy $(x+y)\cdot z=x\cdot z+y\cdot z$ for all $x,y,z\in S$, a zero element and a unit element $0,1\in S$ (i.e. $0\cdot x=0$ and $1\cdot x=x$ for all $x$), and a transitive and reflexive relation (preorder) $\preorderle\subseteq S\times S$. For every $x,y,z\in S$ the $x\preorderle y$ preorder implies that $x+z\preorderle y+z$ and $x\cdot z\preorderle y\cdot z$. Hereinafter we use the same $+,\cdot,0,1$ symbols for the denotation of binary operations and neutral elements (with the multiplication sign often omitted as usual). As the operations and the preorder is usually clear from the context, we will simply denote the preordered semiring with the symbol of the underlying set.

Two preordered semirings will play a distinguished role: the first is the set $\nonnegativereals$ of nonnegative real numbers with its usual addition, multiplication and total order; the second is the tropical semiring. In the multiplicative picture, as a set, the tropical real semiring is the set of nonnegative real numbers $\tropicals=\nonnegativereals$, the sum of $x$ and $y$ is defined as $\max\{x,y\}$, while $\cdot$ is the usual multiplication. Equipped with the usual total order of the real numbers this set is a preordered semiring.

A pair of additional conditions must hold true for the preordered semirings considered here. First the $\naturals\to S$ canonical map which sends $n$ to the $n$-term sum $1+1+\cdots 1$ should be an order embedding (i.e. injective and $m\le n$ as natural numbers iff their images, also denoted by $m$ and $n$, satisfy $m\preorderle n$). We require \emph{polynomial growth} \cite{fritz2020generalization}. A semiring is of polynomial growth if there exist a $u\in S$ \emph{power universal} element such that $u\preorderge 1$ and for every nonzero $x\in S$ there is a $k\in\naturals$ such that $x\preorderle u^k$ and $1\preorderle u^kx$. The power universal element is not necessarily unique, but it can be shown that the subsequent definitions do not depend on a particular choice.

\begin{definition}
Let $S$ be a preordered semiring of polynomial growth and $u\in S$ power universal. The asymptotic preorder is defined by $x\asymptoticge y$ if there is a sequence $(k_n)_{n\in\naturals}$ of natural numbers such that $\lim_{n\to\infty}k_n/n=0$ and for all $n$ the inequality $u^{k_n}x^n\preorderge y^n$ holds.
\end{definition}
\begin{definition}
Let $S$ be a preordered semiring and let $x,y\in S$. If $\exists a\in S$ such that $ax\preorderge ay$ we say that $x$ is catalytically larger than $y$, in notation $x\catalyticge y$.
\end{definition}

\begin{proposition}
$x\preorderge y\implies x\catalyticge y\implies x\asymptoticge y$.
\end{proposition}
\begin{proof}
The first implication is obvious. For the second implication consider $ax\preorderge ay$. Then there exist $k_1\in\naturals$ such that $u^{k_1}\preorderge a$ and $k_2\in\naturals$ such that $u^{k_2}a\preorderge 1$. Thus $u^{k_1+k_2}x\preorderge u^{k_2}ax\preorderge u^{k_2}ay\preorderge y$. In fact there is a constant power realizing the asymptotic ordering. 
\end{proof}

 \begin{definition}
 A $\varphi:S_1\to S_2$ map is \emph{homomorphism} between the semirings $(S_1,\preorderle_1)$ and $(S_2,\preorderle_2)$ if $\varphi(0)=0$, $\varphi(1)=1$, $\varphi(x+y)=\varphi(x)+\varphi(y)$ and $\varphi(xy)=\varphi(x)\varphi(y)$, for every $x,y\in S_1$. If furthermore $x\preorderle_1 y\implies\varphi(x)\preorderle_2\varphi(y)$ for $x,y\in S_1$, then we say that $\varphi$ is a motone semiring homomorphism.
 \end{definition}
 We will be particularly interested in monotone homomorphisms into the real and tropical real semirings. For these we introduce the following notations: given a preordered semiring of polynomial growth $(S,\preorderle)$ with power universal $u$ we let $\testspectrum(S,\preorderle)=\Hom(S,\nonnegativereals)\cup\setbuild{f\in\Hom(S,\tropicals)}{f(u)=2}$ \cite{fritz2020local} and we will call it the 1-test spectrum. Note that in \cite{fritz2020local} monotone decreasing maps are also part of the 1-test spectrum. In our case these parts will be empty, since relative submajorization defined in Section~\ref{sec:familysemiring} will assure that $0\preorderle 1$. While there is a natural normalization condition in the definition of a homomorphism into the nonnegative reals, in the tropical case homomorphisms can always be rescaled in a multiplicative sense by replacing $f(x)$ with $f^c(x)$ for some $c>0$ (see also \cite[Section 13.]{fritz2020local}). This is the reason for requiring that $f(u)=2$ in our definition and the number $2$ itself is arbitrary, but will be convenient relative to our choice of the power universal element $u$ later.

Our strategy will be to use the elements of the spectrum to characterize the catalytic preorder. The main tool will be the following result from \cite{fritz2020local}.
\begin{theorem}[{\cite[second part of 1.4. Theorem, special case]{fritz2020local}}]\label{thm:localglobal}
Let $S$ be a preordered semiring of polynomial growth with $0\preorderle 1$. Suppose that $x,y\in S\setminus\{0\}$ such that for all $f\in\testspectrum(S,\preorderle)$ the strict inequality $f(x)>f(y)$ holds. Then also the following hold:
\begin{enumerate}
\item there is a $k\in\naturals$ such that $u^kx^n\preorderge u^ky^n$ for every sufficiently large $n$
\item if in addition $x$ is power universal then $x^n\preorderge y^n$ for every sufficiently large $n$
\item there is a nonzero $a\in S$ such that $ax\preorderge ay$.
\end{enumerate}
\end{theorem}
In \cite{fritz2020local} a catalyst is given explicitly in terms of the $k$ above. We note that any of the listed conditions implies the non-strict inequalities $f(x)\ge f(y)$ for the monotone homomorphisms.

The following proposition will be useful when dealing with the logarithm of elements of the spectrum.
\begin{proposition}
\label{prop:spectralpointpositive}
Let $S$ be a preordered semiring of polynomial growth and $u\in S$ power universal. Then for every monotone homomorphism $f$ from $S$ into either $\reals$ or $\tropicals$ and for every nonzero $x\in S$ we have that $f(x)> 0$.
\end{proposition}
\begin{proof}
Since $u$ is power universal, $u\geq 1$ and thus $f(u)\geq f(1)=1$. Then for every nonzero $x\in S$ there is a $k\in \naturals$ such that $1\preorderle u^k x$. This yields $f(u)^kf(x)\geq f(1)=1$ and $f(x)\geq f(1)f(u)^{-k}>0$.
\end{proof}
Observe that the inequality $f(u)\ge 1$ can be strengthened as follows. On the one hand, by the chosen normalization of tropical monotones we have $f(u)\ge 2>1$. On the other hand, by power universality, there is a $k\in\naturals$ such that $u^k\preorderge 2$. Apply the real monotone homomorphism $f$ and rearrange to get $f(u)\preorderge 2^{1/k}>1$. This allows us to show that the asymptotic relaxation of the preorder holds even if we only have non-strict inequalities on the spectrum.
\begin{corollary}\label{cor:nonstrictasymptotic}
Let $S$ be a preordered semiring of polynomial growth and $u\in S$ power universal. Then $x\asymptoticge y$ iff for all $f$ in the 1-test spectrum the inequality $f(x)\ge f(y)$ holds. 
\end{corollary}
\begin{proof}
The only if direction is clear: $u^{k_n}x^n\preorderge y^n$ implies $f(u)^{k_n/n}f(x)\ge f(y)$, and by taking the limit as $n\to\infty$, also $f(x)\ge f(y)$. For the if direction, recall that $f(u)>1$ for all $f$ in the spectrum. Assuming $f(x)\ge f(y)$ for all $f$, this implies that for all $n\in\naturals$ we have the strict inequalities $f(ux^n)>f(y^n)$. By Theorem~\ref{thm:localglobal}, there exists nonzero $a\in S$ (which may depend on $n$), such that $aux^n\preorderge ay^n$. By \cite[Lemma 2 (iv)]{vrana2021generalization}, this implies $ux^n\asymptoticge y^n$ for all $n$, which in turn by \cite[Lemma 3]{vrana2021generalization} implies $x\asymptoticge y$.
\end{proof}

\section{Relative submajorization of state families}\label{sec:relsubfamilies}

\subsection{The preordered semiring of pairs of families}\label{sec:familysemiring}

Let $X,Y$ be nonempty compact Hausdorff topological spaces. These will be the index sets for the families, and can be considered fixed throughout this section. We will consider pairs of continuous maps $(\rho,\sigma)$, where $\rho:X\to\positivedefiniteoperators{\mathcal{H}}$ and $\sigma:Y\to\positivedefiniteoperators{\mathcal{H}}$ for some finite dimensional Hilbert space $\mathcal{H}$. Two pairs $(\rho,\sigma)$ and $(\rho',\sigma')$ are equivalent if there is a unitary $U:\mathcal{H}\to\mathcal{H}'$ such that $\forall x\in X:U\rho(x)U^*=\rho'(x)$ and $\forall y\in Y:U\sigma(y)U^*=\rho'(y)$. We let $\familysemiring{X}{Y}$ denote the set of equivalence classes of pairs of such families. The pointwise direct sum and tensor product operations are well-defined on equivalence classes, and turn $\familysemiring{X}{Y}$ into a commutative semiring. The zero element is represented by the unique pair over a zero dimensional Hilbert space, while $1$ is represented by the pair consisting of constant functions with value $I$ over $\complexes$.

We adopt the convention that if any map $\rho:X\to\positivedefiniteoperators{\mathcal{H}}$ and $\sigma:Y\to\positivedefiniteoperators{\mathcal{H}}$ appear outside of brackets, then any operations or relations they appear in, are to be understood pointwise, including sum, product, direct sum, tensor product or image under a linear super operator, usually notated by $T$. More precisely, given any map $\positivedefiniteoperators{\mathcal{H}}\to\positivedefiniteoperators{\mathcal{H'}}$ we understand the composition $T(\rho):x\mapsto T(\rho(x))$ and $T(\sigma):y\mapsto T(\sigma(y))$.
\begin{definition}
\label{def:preorder}
$(\rho,\sigma)$ relatively submajorizes $(\rho',\sigma')$, in notation $(\rho,\sigma)\preorderge(\rho',\sigma')$, if there exists a completely positive trace-nonincreasing map $T:\boundeds(\mathcal{H})\to\boundeds(\mathcal{H}')$ such that $T(\rho)\ge\rho'$ and $T(\sigma)\le\sigma'$.
\end{definition}

\begin{proposition}
$\familysemiring{X}{Y}$ is a preordered semiring with relative submajorization.
\end{proposition}
\begin{proof}
We need to verify that the preorder is compatible with the semiring operations. Suppose that $(\rho,\sigma)\preorderge(\rho',\sigma')$ and let $T$ be a completely positive trace non-increasing map as in Definition~\ref{def:preorder}. Let $(\omega,\tau)\in\familysemiring{X}{Y}$ be a pair of families on $\mathcal{K}$, ie. $\omega:X\to\positivedefiniteoperators{\mathcal{K}}$ and $\tau:Y\to\positivedefiniteoperators{\mathcal{K}}$. Then
\begin{equation}
\begin{aligned}
(T\otimes\id_{\boundeds(\mathcal{K})})(\rho\otimes\omega) & = T(\rho)\otimes\omega \ge \rho'\otimes\omega  \\
(T\otimes\id_{\boundeds(\mathcal{K})})(\sigma\otimes\tau) & = T(\sigma)\otimes\tau \le \sigma'\otimes\tau,
\end{aligned}
\end{equation}
therefore $(\rho,\sigma)(\omega,\tau)\preorderge(\rho',\sigma')(\omega,\tau)$.

The map $\tilde{T}:\boundeds(\mathcal{H}\oplus\mathcal{K})\to\boundeds(\mathcal{H}'\oplus\mathcal{K})$ defined as
\begin{equation}
\tilde{T}\left(\begin{bmatrix}
A & B  \\
C & D
\end{bmatrix}
\right)=\begin{bmatrix}
T(A) & 0  \\
0 & D
\end{bmatrix}
\end{equation}
is also completely positive and trace non-increasing, and satisfies
\begin{equation}
\begin{aligned}
\tilde{T}(\rho\oplus\omega) & = T(\rho)\oplus\omega \ge \rho'\oplus\omega\\
\tilde{T}(\sigma\oplus\tau) & = T(\sigma)\oplus\tau \le \sigma'\oplus\tau,
\end{aligned}
\end{equation}
therefore $(\rho,\sigma)+(\omega,\tau)\preorderge(\rho',\sigma')+(\omega,\tau)$.
\end{proof}

\begin{proposition}
$\familysemiring{X}{Y}$, is of polynomial growth and $u=(2,1)$ is a power universal. 
\end{proposition}
\begin{proof}
For the pair of families $(\rho,\sigma)$ let us choose the substochastic map $T(.):=c_1\Tr(.)$, with $c_1=\min\{1,[\max_{y\in Y}\Tr \sigma(y)]^{-1}$\} . Then we have $c_1\Tr \sigma\leq 1$ and $T(u^k(\rho,\sigma))=(T(2^k\rho),T(\sigma))=(c_12^k\Tr\rho,c_1\Tr\sigma)$. Choosing a large enough $k$ will satisfy $c_12^k\Tr\rho\geq 1$, since $\Tr \rho$ is bounded on $X$ and so $\exists k\in\mathbb{N}:u^k(\rho,\sigma)\preorderle (1,1)$.
Let us choose now $T(.):=c_2(.)\otimes\frac{1_d}{d}$, with $c_2=\min\{1,d[\min_{y\in Y}\min(\spectrum(\sigma(y))]\}$. Then we have $\sigma\geq \frac{c_2}{d}1_d$ and $T(u^k)=(T(2^k),T(1))=(\frac{2^kc_2}{d}1_d,\frac{c_2}{d}1_d)$. Choosing a large enough $k$ will satisfy $\frac{2^kc_2}{d}1_d\geq \rho$, since $\max\spectrum(\rho)$ is bounded on $X$ and so $\exists k\in\mathbb{N}:u^k\preorderge (\rho,\sigma)$.
\end{proof}

$\familysemiring{X}{Y}$ is then a semiring of polynomial growth and in $\familysemiring{X}{Y}$ we have $0\preorderle 1$ and thus Theorem~\ref{thm:localglobal} and Corollary~\ref{cor:nonstrictasymptotic} are applicable.

\subsection{Classical families}\label{sec:classical}

\begin{definition}
The subsemiring generated by the one-dimensional elements is called the subsemiring of classical families, in notation $\cfamilysemiring{X}{Y}$. That is $(\rho,\sigma)\in\cfamilysemiring{X}{Y}$ if and only if $[\rho(x),\rho(x)]=[\rho(x),\sigma(y)]=[\sigma(y),\sigma(y)]=0,\;\forall x\in X,\;\forall y\in Y$. 
\end{definition}

We turn to the classification of real and tropical real valued monotone homomorphisms on the subsemiring of classical families. By definition every element in $\cfamilysemiring{X}{Y}$ is a sum of one-dimensional elements. A one-dimensional element of the semiring on the other hand can be identified by a pair of strictly positive continuous functions on $X$ and $Y$. Suppose that $f$ is a multiplicative map from the one-dimensional pairs into either the real or the tropical numbers. Then the extension of $f$ to multi-dimensional pairs via additivity also enjoys multiplicativity. Since $\cfamilysemiring{X}{Y}$ is generated by the one-dimensional pairs, the value of every $f\in \testspectrum(\cfamilysemiring{X}{Y})$ is determined by its behaviour on one-dimensional pairs.

\begin{proposition}
\label{prop:integralform}
If $f\in \testspectrum(\cfamilysemiring{X}{Y})$ then there exists unique, non-negative Radon measures $\mu$ and $\nu$ on $X$ and $Y$ such that for every multidimensional classical pair $(\bigoplus_{i=1}^dp_i,\bigoplus_{i=1}^dq_i)$ ($p_i\in C(X),q_i\in C(Y)$), if $f$ is real valued, it admits the form
\begin{align}\label{eq:realmonotones}
    \sum_{i=1}^d \exp\left(\int_X\ln p_i\ed\mu-\int_Y\ln q_i\ed\nu\right),\qquad\mu(X)-\nu(Y)=1,
\end{align}
while if $f$ is tropical valued it admits the form
\begin{align}\label{eq:tropicalmonotones}
    \max_{i\in[d]} \exp\left(\int_X\ln p_i\ed\mu-\int_Y\ln q_i\ed\nu\right),\qquad\mu(X)-\nu(Y)=0
\end{align}
and functions of these form are monotone under relative majorization if and only if they satisfy the data processing inequality.
\end{proposition}

\begin{proof}
 Let $f\in\testspectrum(\cfamilysemiring{X}{Y})$ be an element of the spectrum. For every $\xi,\eta>0$ one has $(e^\xi,1_Y)\geq (1_X,1_Y)$ and $(1_X,e^{-\eta})\geq (1_X,1_Y)$, thus the maps $\xi\mapsto\ln f(e^\xi,1_Y)$, from $C(X)$ to $\mathbb{R}$ and $\eta\mapsto\ln f(1_X,e^{-\eta})$, from $C(Y)$ to $\mathbb{R}$ are well defined positive linear functionals on $C(X)$ and $C(Y)$ (note that we can take the logarithm by Proposition~\ref{prop:spectralpointpositive}). Thus by Theorem~\ref{thm:Riesz}, $\ln f(e^\xi,1_Y)=\int_X\xi(x)\ed\mu(x)$ and $\ln f(1_Y,e^{-\eta})=\int_Y\eta(y)\ed\nu(y)$ for some unique $\mu,\nu$ Radon measures on $X$ and $Y$. Since $f$ is multiplicative $\ln f(e^\xi,e^{-\eta})=\ln f(e^\xi,1_Y)+\ln f(1_X,e^{-\eta})=\int_X\xi\ed\mu+\int_Y\eta\ed\nu$. From this $ f(p,q)=\exp(\int_X\ln p\ed\mu-\int_Y\ln q\ed\nu)$. We used only the multiplicative property of $f$ but not the additive property, thus this part of the proof works for either real or tropical valued elements of the spectrum. Consider now $f(1_X+1_X,1_Y+1_Y)=f(1_X,1_Y)+f(1_X,1_Y)$. In the real case this translates to $f(1_X+1_X,1_Y+1_Y)=2$, in the tropical case to $f(1_X+1_X,1_Y+1_Y)=1$. Then $t\mapsto \ln f(e^t1_X,e^t1_Y)$ is additive, normalized and monotone, therefore satisfies Cauchy's functional equation and admits the form $\ln f(e^t1_X,e^t1_Y)=t$ in the real case and the form $\ln f(e^t1_X,e^t1_Y)=0$ in the tropical case. This leads to $1=\ln f(e1_X,e1_Y)=\mu(X)-\nu(Y)$ in the real case and $0=\ln f(e1_X,e1_Y)=\mu(X)-\nu(Y)$ in the tropical case. This further shows that elements of the spectrum are homogeneous of degree $1$ in the real case and homogeneous of degree $0$ in the tropical case.
 
 Now from additivity any real or tropical valued element of the spectrum admits the forms \eqref{eq:realmonotones} and \eqref{eq:tropicalmonotones}.
 We fully exploited additivity and multiplicativity, we further need to impose monotonocity under relative submajorization on multidimensional pairs of families. Elements of the spectrum need to be monotone under relative submajorization, so in particular relative majorization. Functions of the form \eqref{eq:realmonotones} and \eqref{eq:tropicalmonotones} are monotone decreasing under increase of any of the $q_i$ or decrease of any of the $p_i$. So these functions are elements of the spectrum if and only if they are monotone decreasing under relative majorization, ie. under stochastic maps, classical channels, that is we further require \eqref{eq:realmonotones} and \eqref{eq:tropicalmonotones} to satisfy the data-processing inequality.
\end{proof}

\begin{lemma}
\label{lem:jointconvexity}
Let $f$ be an additive function from $\familysemiring{X}{Y}$ into either the real or tropical numbers. Then
\begin{enumerate}
    \item if $f$ is homogeneous of degree $1$ and $f$ goes into the real numbers, then it satisfies the data-processing inequality if and only if it is jointly convex;
    \item if $f$ is homogeneous of degree $0$ and $f$ goes into the tropical numbers, then it satisfies the data-processing inequality if and only if it is jointly quasi-convex.
\end{enumerate}
\end{lemma}
\begin{proof}
Let $f$ be an additive function from $\familysemiring{X}{Y}$ into either the real or tropical numbers and let $f$ be homogeneous of degree $k$. Whenever the $\sum$ symbol is outside of $f$ let it stand as summing with respect to the semiring: usual summing in the real case and maximum in the tropical case.
Suppose $f$ is monotone under quantum channels. Applying monotonocity to $\hat{\rho}:=\sum_i p_i\ketbra{i}{i}_E\otimes\rho_i$ and $\hat{\sigma}:=\sum_i p_i\ketbra{i}{i}_E\otimes\sigma_i$ under the partial trace $\Tr_E$, where $\left(\ket{i}\right)_{i=1}^r$ is an ONS in $\mathcal{H}_E$ yields
\begin{align}
\sum_i p_i^k f\left(\rho_i,\sigma_i\right)=f\left(\sum_i p_i\ketbra{i}{i}_E\otimes\rho_i,\sum_i p_i\ketbra{i}{i}_E\otimes\sigma_i\right)\geq f\left(\sum_i p_i\rho_i,\sum_i p_i\sigma_i\right),
\end{align}
which translates to joint convexity when $k=1$, ie. in the real case and joint quasi-convexity, when $k=0$, ie in the tropical case.
Suppose now that $f$ is jointly convex in the real case or jointly quasi-convex in the tropical case.
Using Stinespring dilation $\Phi\left(.\right)=\Tr_E V\left(.\right)V^*$ with an isometry $V:\mathcal{H}\to\mathcal{K}\otimes\mathcal{H}_E$, and writing the partial trace multiplied by the maximally mixed state as a convex combination of unitary conjugations (e.g., by the discrete Weyl unitaries):
\begin{equation}
\begin{split}
f\left(\Phi\left(\rho\right),\Phi\left(\sigma\right)\right)&=\sum_{i=1}^{d_E}\frac{1}{d_E^k}f\left(\Phi\left(\rho\right),\Phi\left(\sigma\right)\right)\\
&=f\left(\frac{1}{d_E}I_E\otimes\Tr_E V\left(\rho\right)V^*,\frac{1}{d_E}I_E\otimes\Tr_E V\left(\sigma\right)V^*\right)\\
&=f\left(\sum_i\frac{1}{n} U_i V\rho V^*U_i^*,\sum_i\frac{1}{n} U_i V\sigma V^*U_i^*\right)\\
&\leq \frac{1}{n^k}\sum_{i=1}^n f\left(\rho,\sigma\right)=f\left(\rho,\sigma\right).
\end{split}
\end{equation}
\end{proof}

\begin{proposition}
\label{prop:concentrated}
If functions of the form \eqref{eq:realmonotones} and \eqref{eq:tropicalmonotones} satisfy the data processing inequality then the measure $\mu$ is concentrated on one point.
\end{proposition}
\begin{proof}
By Lemma~\ref{lem:jointconvexity} we require functions of the form \eqref{eq:realmonotones} and \eqref{eq:tropicalmonotones} to be jointly convex and jointly quasi-convex. In particular both family of functions needs to be jointly quasiconvex in the one dimensional special case. These functions are totally differentiable and if we restrict $f$ to a line segment then having a zero directional derivative and negative second derivative would mean strict local maximum and would contradict quasiconvexity. Consider the general directional derivative of  $f$ at $1$. The forms of $f$ in \eqref{eq:realmonotones} and \eqref{eq:tropicalmonotones} are differentiable and the derivative of the integrands are continuous on $X$ and $Y$ and thus bounded. Then by \cite[Theorem 2.27]{folland1999real} the differentiation and the integration commute. 
\begin{equation}
    \begin{split}
        \left.\frac{\ed}{\ed s}f(1_X+s\xi,1_Y)\right|_{s=0}
         & =\left.\frac{\ed}{\ed s}\left[\exp\int_X(\ln 1_X+s\xi)\ed\mu \right]\right|_{s=0}\\
         & =\left.f(1_X+s\xi,1_Y)\left[\int_X\frac{\xi}{1_X+s\xi} \ed\mu \right]\right|_{s=0}=\int_X\xi \ed\mu 
    \end{split}
\end{equation}
and
\begin{equation}
    \begin{split}
        \left.\frac{\ed^2}{\ed s^2}f(1_X+s\xi,1_Y)\right|_{s=0}
         & =\left.\frac{\ed^2}{\ed s^2}\left[\exp\int_X(\ln 1_X+s\xi)\ed\mu \right]\right|_{s=0}\\
         & =\left.f(1_X+s\xi,1_Y)\left[\left(\int_X\frac{\xi}{1_X+s\xi} \ed\mu \right)^2-\int_X\left(\frac{\xi}{1_X+s\xi} \right)^2\ed\mu \right]\right|_{s=0}\\
         & =\left(\int_X\xi\ed\mu\right)^2-\int_X\xi^2\ed\mu.
    \end{split}
\end{equation}
To get a contradiction with quasiconvexity, we need to find a continuous function $\xi$ whose integral is zero and such that $\xi^2$ has nonzero integral, supposing that $\mu$ is not concentrated on a point. To this end let $x_1,x_2$ be distinct points in the support of $\mu$. Choose disjoint closed neighbourhoods $A_1$ and $A_2$ (possible since $X$ is a compact Hausdorff space). By Urysohn's lemma, there exist continuous functions $\xi_1,\xi_2:X\to[0,1]$ such that $\xi_1$ is $1$ on $A_1$ and $0$ on $A_2$, and $\xi_2$ is $0$ on $A_1$ and $1$ on $A_2$. Let
\begin{equation}
    \xi=\left(\int_X\xi_2\ed\mu\right)\xi_1-\left(\int_X\xi_1\ed\mu\right)\xi_2.
\end{equation}
By construction, the integral of $\xi$ vanishes, while
\begin{equation}
    \int_X\xi^2\ed\mu\ge\mu(A_1)\left(\int_X\xi_2\ed\mu\right)^2>0.
\end{equation}
\end{proof}
Let $\alpha:=\mu(X)$ and $\gamma:=\frac{\nu}{\nu(Y)}$. Taking Propositions~\ref{prop:integralform} and~\ref{prop:concentrated} into account an element of the 1-test spectrum must have one of the following forms.
\begin{align}\label{eq:realmonotonespretheorem}
f\left(\bigoplus_{i=1}^dp_i,\bigoplus_{i=1}^dq_i\right)=\sum_{i=1}^d p_i(x)^\alpha\exp\left[(1-\alpha)\int_Y\ln q_i\ed\gamma\right]
\end{align}
if $f$ goes into the reals and
\begin{align}\label{eq:tropicalmonotonespretheorem}
f\left(\bigoplus_{i=1}^dp_i,\bigoplus_{i=1}^dq_i\right)=\max_{i\in[d]} p_i(x)\exp\left[-\int_Y\ln q_i\ed\gamma\right],
\end{align}
if $f$ goes into the tropicals, where we further took into account that $f(u)=2$ is required for elements of the spectrum going into the tropical numbers. \eqref{eq:realmonotonespretheorem} is characterized by the point $x$, the weight $\alpha\geq 1$ and the probability measure $\gamma$, while \eqref{eq:tropicalmonotonespretheorem} is characterized by the point $x$ and the probability measure $\gamma$.
\begin{proposition}\label{prop:satisfyDPI}
\eqref{eq:realmonotonespretheorem} and \eqref{eq:tropicalmonotonespretheorem} satisfy the data processing inequality.
\end{proposition}
\begin{proof}
 By Lemma~\ref{lem:jointconvexity} the functions \eqref{eq:realmonotonespretheorem} are monotone under channels if and only if they are jointly convex. Since the sum of such functions also enjoy the joint convexity, it is sufficient to show joint convexitiy in the 1-dimensional case. Joint convexitiy then equivalently translates to the second directional derivative being nonnegative at any point in any direction. Let us compute the general directional second derivative relying on the commutativity of the differentiation and integration by \cite[Theorem 2.27]{folland1999real}. 
 \begin{align}
     f(p+s\xi,q+s\eta)=f(p,q)f(1_X+s\tilde{\xi},1_Y+s\tilde{\eta}),
 \end{align}
 with $\tilde{\xi}:=\frac{\xi}{p}$ and $\tilde{\eta}:=\frac{\eta}{q}$. Then as in Proposition~\ref{prop:concentrated}
\begin{equation}
    \begin{split}
    \left.\frac{\ed^2}{\ed s^2}f(p+s\xi,q+s\eta)\right|_{s=0}
    &=\left.f(p,q)\frac{\ed}{\ed s}f(1_X+s\tilde{\xi},1_Y+s\tilde{\eta})\left[\alpha\frac{\tilde{\xi}(x)}{1_X+s\tilde{\xi}(x)}+(1-\alpha)\int_Y\frac{\tilde{\eta}}{1_Y+s\tilde{\eta}} \ed\gamma \right]\right|_{s=0}\\
    &=f(p,q)\left[\left(\alpha\tilde{\xi}(x)+(1-\alpha)\int_Y\tilde{\eta} \ed\gamma\right)^2-\alpha\tilde{\xi}(x)^2-(1-\alpha)\int_Y\tilde{\eta}^2 \ed\gamma \right]  \\
    & = f(p,q)(\alpha-1)\left[\alpha\tilde{\xi}(x)^2-2\alpha\tilde{\xi}(x)\int_Y\tilde{\eta}\ed\gamma+\int_Y\tilde{\eta}^2\ed\gamma+(\alpha-1)\left(\int_Y\tilde{\eta}\ed\gamma\right)^2\right]  \\
    & \ge f(p,q)(\alpha-1)\alpha\left[\tilde{\xi}(x)^2-2\tilde{\xi}(x)\int_Y\tilde{\eta}\ed\gamma+\left(\int_Y\tilde{\eta}\ed\gamma\right)^2\right]  \\
    & = f(p,q)(\alpha-1)\alpha\left(\tilde{\xi}(x)-\int_Y\tilde{\eta}\ed\gamma\right)^2\geq 0,
    \end{split}
\end{equation}
where we used that the second moment of $\tilde{\eta}$ is greater than the square of the first moment. We conclude that the functions \eqref{eq:realmonotonespretheorem} are jointly convex and thus satisfy the data processing inequality, they are monotone decreasing under stohastic maps. Now for a function $f$ of the form \eqref{eq:realmonotonespretheorem} consider
\begin{equation}
\begin{split}
g_{\alpha}\left(\bigoplus_{i=1}^dp_i,\bigoplus_{i=1}^dq_i\right)&:=\left(\sum_{i=1}^d f\left(p_i,q_i\right)\right)^{\frac{1}{\alpha}}\\
&=\left(\sum_{i=1}^d p_i^\alpha(x)\exp\left[\left(1-\alpha\right)\int_Y\ln q_i\ed\gamma\right]\right)^{\frac{1}{\alpha}}.
\end{split}
\end{equation}
$g_\alpha$ then satisfies the data processing inequality and preserves this property in a $\alpha\to\infty$ limit. However
\begin{align}
    \lim_{\alpha\to\infty}g_{\alpha}\left(\bigoplus_{i=1}^dp_i,\bigoplus_{i=1}^dq_i\right)=\max_{i\in[d]}f\left(p_i,q_i\right)=\max_{i\in[d]} p_i(x)\exp\left[-\int_Y\ln q_i\ed\gamma\right],
\end{align}
showing that functions of the form \eqref{eq:tropicalmonotonespretheorem} satisfy the data processing inequality too.
\end{proof}
Note that functions of the form \eqref{eq:realmonotonespretheorem} can be viewed as the $\alpha$-Rényi quasidivergences of a positive vector $p$ and some pointwise geometric mean of positive vectors $q_i$. What we used in the last proof is that the max divergence, ie. functions in \eqref{eq:tropicalmonotonespretheorem}, can be given as a $\alpha\to\infty$ limit of Rényi divergences.

\begin{theorem}\label{thm:realspectrum}
$\testspectrum(\cfamilysemiring{X}{Y}$) consists of the functions
\begin{align}\label{eq:realmonotonestheorem}
f_{\alpha,x,\gamma}\left(\bigoplus_{i=1}^dp_i,\bigoplus_{i=1}^dq_i\right)=\sum_{i=1}^d p_i(x)^\alpha\exp\left[(1-\alpha)\int_Y\ln q_i\ed\gamma\right],
\end{align}
where $\alpha\geq 1$, $x\in X$ and $\gamma$ is a probablity measure on $Y$, if $f$ is real-valued and
\begin{align}\label{eq:tropicalmonotonestheorem}
f_{x,\gamma}\left(\bigoplus_{i=1}^dp_i,\bigoplus_{i=1}^dq_i\right)=\max_{i\in[d]} p_i(x)\exp\left[-\int_Y\ln q_i\ed\gamma\right],
\end{align}
where, $x\in X$ and $\gamma$ is a probability measure on $Y$, if $f$ is tropical real-valued.
\end{theorem}
\begin{proof}
Follows from Propositions~\ref{prop:integralform},~\ref{prop:concentrated} and~\ref{prop:satisfyDPI}.
\end{proof}

We will be refering to elements of $\testspectrum(\cfamilysemiring{X}{Y})$ as $f_{(\alpha),x,\gamma}$ or $g_{(\alpha),x,\gamma}$, if we want distinct multiple elements of the spectrum, signifying the characterizing quantities $x,\gamma$ and possibly $\alpha$ as well, the same time.

\subsection{Quantum extensions}\label{sec:quantum}

\begin{proposition}\label{prop:singelepointquantum}
Suppose that $\tilde{f}$ is an element of the spectrum. Let $f_{(\alpha),x,\gamma}$ be $\tilde{f}$ constrained on the classical semiring according to \eqref{eq:realmonotonestheorem} or \eqref{eq:tropicalmonotonestheorem} in Theorem~\ref{thm:realspectrum}. Then for any $\rho,\rho':X\to \positivedefiniteoperators{\mathcal{H}}$ such that $\rho(x)=\rho'(x)$ and for any $\sigma:Y\to\positivedefiniteoperators{\mathcal{H}}$ it follows that $\tilde{f}(\rho,\sigma)=\tilde{f}(\rho',\sigma)$.
\end{proposition}
\begin{proof}
Let $c_1(x)=\min\{t|t\rho(x)\ge\rho'(x)\}$ and $c_2(x)=\min\{t|t\rho'(x)\ge\rho(x)\}$. Then $c_1$ and $c_2$ are strictly positive continuous functions on $X$, and
\begin{align}
    (\rho',\sigma)\leq(c_1,1)(\rho,\sigma)\leq(c_1,1)(c_2,1)(\rho',\sigma).
\end{align}
$c_1(x)=c_2(x)=1$, $(c_1,1)$ and $(c_2,1)$ are classical pairs and thus $\tilde{f}(c_1,1)=f_{(\alpha),x,\gamma}(c_1,1)=f_{(\alpha),x,\gamma}(1,1)=1$ and $\tilde{f}=f_{(\alpha),x,\gamma}(c_2,1)=f_{(\alpha),x,\gamma}(1,1)=1$. Applying now $\tilde{f}$ to all three parts of the above inequality yields
\begin{align}
    \tilde{f}(\rho',\sigma)\leq \tilde{f}(\rho,\sigma)\leq \tilde{f}(\rho',\sigma).
\end{align}
\end{proof}

\begin{proposition}\label{prop:commutingsigma}
Let $(\rho,\sigma)\in\familysemiring{X}{Y}$ and let $\tilde{f}$ be a real element of the spectrum and let $f_{\alpha,x,\gamma}$ be its restriction onto the classical subsemiring. Let $\tilde{g}$ be a tropical element of the spectrum and let $g_{x,\gamma}$ be its restriction onto the classical subsemiring. If $[\sigma(y),\sigma(y')]=0\;\forall y,y'\in Y$  then
\begin{equation}
    \tilde{f}\left(\rho,\sigma\right) =\sandwichedquasientropy{\alpha}{\rho(x)}{\exp\int_Y \ln\sigma \ed\gamma}
\end{equation}
and
\begin{equation}
    \tilde{g}\left(\rho,\sigma\right) =\norm[\infty]{\rho^{\frac{1}{2}}(x)\left(\exp\int_Y \ln\sigma \ed\gamma\right)^{-1}\rho^{\frac{1}{2}}(x)}.
\end{equation}
\end{proposition}
\begin{proof}
There is a positive definite operator $\tilde{\sigma}$ such that the eigenbasis of $\tilde{\sigma}$ simultaneously diagonalizes all $\sigma(y)$. Let $\mathcal{P}_{\tilde{\sigma}_n}$ denote the pinching by $\tilde{\sigma}^{\otimes n}$, then $\mathcal{P}_{\tilde{\sigma}_n}$ leaves $\sigma^{\otimes n}(y)$ invariant for all $y\in Y$. It follows that
\begin{equation}
\begin{split}
\left(\rho,\sigma\right)^{n}&\preorderge\left(\mathcal{P}_{\tilde{\sigma}_n}\left(\rho^{\otimes n}\right),\mathcal{P}_{\tilde{\sigma}_n}\left(\sigma^{\otimes n}\right)\right)\\
&=\left(\mathcal{P}_{\tilde{\sigma}_n}\left(\rho^{\otimes n}\right),\sigma^{\otimes n}\right)\\
&\preorderge\left(\frac{1}{\poly(n)},1\right)\left(\rho^{\otimes n},\sigma^{\otimes n}\right),
\end{split}
\end{equation}
where $\poly(n)$ is a polynomial of $n$ and we used that any pinching is a completely positive trace preserving map and the pinching inequality:
\begin{equation}
\rho^{\otimes n}\leq |\spectrum(\tilde{\sigma}^{\otimes n})|\mathcal{P}_{\tilde{\sigma}_n}\left(\rho^{\otimes n}\right)=\poly(n)\mathcal{P}_{\tilde{\sigma}_n}\left(\rho^{\otimes n}\right).
\end{equation}
We have shown in Proposition~\ref{prop:singelepointquantum} that $\tilde{f}$ and $\tilde{g}$ only depend on one point of $\rho$ apart from $\sigma$, but after the pinching all these operators commute and thus we are evaluating $\tilde{f}$ and $\tilde{g}$ on the classical subsemiring, where $\tilde{f}$ and $\tilde{g}$ are determined by Theorem~\ref{thm:realspectrum}. Applying $f$ to all three parts and taking the $n$-th root yields
\begin{equation}
\begin{split}
\tilde{f}\left(\rho,\sigma\right)&\geq \sqrt[n]{f_{\alpha,x,\gamma}\left(\mathcal{P}_{\tilde{\sigma}_n}\left(\rho^{\otimes n}(x)\right),\sigma^{\otimes n}\right)}\\
&=\sqrt[n]{\Tr \left(\mathcal{P}_{\tilde{\sigma}_n}\left(\rho^{\otimes n}\right)\right)^\alpha\left(\exp\int_Y \ln\sigma^{\otimes n} \ed\gamma\right)^{1-\alpha}}\\
&=\sqrt[n]{\Tr \left(\mathcal{P}_{\tilde{\sigma}_n}\left(\rho^{\otimes n}(x)\right)\right)^\alpha\left(\left(\exp\int_Y \ln\sigma \ed\gamma\right)^{\otimes n}\right)^{1-\alpha}}\\
&\geq \sqrt[n]{\frac{1}{\left(\poly(n)\right)^\alpha}} \tilde{f}\left(\rho,\sigma\right).
\end{split}
\end{equation}
Taking the limit $n\to +\infty$ gives us
\begin{equation}
\tilde{f}\left(\rho(x),\sigma\right)\geq \sandwichedquasientropy{\alpha}{\rho}{\exp\int_Y \ln\sigma \ed\gamma}\geq \tilde{f}\left(\rho,\sigma\right),
\end{equation}
where we refer to \cite[Proposition 4.12.]{tomamichel2015quantum} (see also \cite[Theorem 4.4.]{perry2020semiring}) in taking the limit of the middle term.

Now applying $\tilde{g}$ to all three parts yields
\begin{equation}
\begin{split}
\tilde{g}\left(\rho,\sigma\right)&\geq \sqrt[n]{g_{x,\gamma}\left(\mathcal{P}_{\tilde{\sigma}_n}\left(\rho^{\otimes n}(x)\right),\sigma^{\otimes n}\right)}\\
&=\sqrt[n]{\norm[\infty]{\left(\mathcal{P}_{\tilde{\sigma}_n}\left(\rho^{\otimes n}(x)\right)\right)\left(\exp\int_Y \ln\sigma^{\otimes n} \ed\gamma\right)}}\\
&=\sqrt[n]{\norm[\infty]{\left(\mathcal{P}_{\tilde{\sigma}_n}\left(\rho^{\otimes n}(x)\right)\right)\left(\left(\exp\int_Y \ln\sigma \ed\gamma\right)^{\otimes n}\right)}}\\
&\geq \sqrt[n]{\frac{1}{\poly(n)}} \tilde{g}\left(\rho,\sigma\right).
\end{split}
\end{equation}
Taking the limit $n\to +\infty$ gives us
\begin{equation}
\tilde{g}\left(\rho,\sigma\right)\geq \norm[\infty]{\rho^{\frac{1}{2}}(x)\left(\exp\int_Y \ln\sigma \ed\gamma\right)^{-1}\rho^{\frac{1}{2}}(x)}\geq \tilde{g}\left(\rho,\sigma\right),
\end{equation}
where we refer to \cite{datta2009min} and \cite[Section 4.2.4]{tomamichel2015quantum} in taking the limit of the middle term.
\end{proof}

The expression $\exp\int_Y\ln\sigma\ed\gamma$ can be viewed as a continuous analogue of a weighted geometric mean of positive numbers. The form of the spectrum elements in the case of commuting $\sigma$ suggests looking for fully quantum generalizations of the form $f(\rho,\sigma)=\sandwichedquasientropy{\alpha}{\rho(x)}{M(\sigma)}$, where $\alpha\ge 1$, $x\in X$ and $M$ is some noncommutative version of the weighted geometric mean. In the following definition we make the requirements more precise and also more flexible by allowing the result to be also a continuous family of positive operators. The advantages of this formulation are that a simple composition property conveniently allows for the construction of many examples, and that these objects also give rise to monotone homomorphisms between different semirings $\familysemiring{X}{Y}\to\familysemiring{X'}{Y'}$. We equip $C(Y,\positivedefiniteoperators{\mathcal{H}})$ with the pointwise semidefinite partial order.
\begin{definition}\label{def:mean}
Let $Y$ and $Y'$ be nonempty compact spaces. A family of continuous geometric means indexed by $Y'$ is a collection of maps $M:C(Y,\positivedefiniteoperators{\mathcal{H}})\to C(Y',\positivedefiniteoperators{\mathcal{H}})$ which are unitary equivariant (i.e. if $\sigma'=U\sigma U^*$ for some unitary $U:\mathcal{H}\to\mathcal{H}'$, then $M(\sigma')=UM(\sigma)U^*$ satisfying the following properties for all $\sigma\in C(Y,\positivedefiniteoperators{\mathcal{H}})$, $\sigma'\in C(Y,\positivedefiniteoperators{\mathcal{H}'})$ and $\lambda\in\positivereals$:
\begin{enumerate}
    \item $M(\sigma\oplus\sigma')=M(\sigma)\oplus M(\sigma')$,
    \item $M(\sigma\otimes\sigma')=M(\sigma)\otimes M(\sigma')$,
    \item $M(\lambda \sigma)=\lambda M(\sigma)$,
    \item if $\sigma\le\sigma'$, then $M(\sigma)\le M(\sigma')$,
    \item $M$ is concave.
\end{enumerate}
The set of families of geometric means is denoted by $\gmeans(Y,Y')$. When $Y'$ is a one-point space, we identify $C(Y',\positivedefiniteoperators{\mathcal{H}})$ with $\positivedefiniteoperators{\mathcal{H}}$ and write $\gmeans(Y)$ instead of $\gmeans(Y,Y')$.
\end{definition}
Because of unitary equivariance, it is sufficient to specify a family of means for families of operators on $\complexes^d$ for all $d$.

We note that the properties of families of geometric means that we consider imply that they can be extended to positive semidefinite operators (by $\lim_{\epsilon\to 0}M(\sigma+\epsilon 1_Y\otimes I_\mathcal{H})$), and that they are increasing under completely positive trace-preserving maps in the sense that if $M\in\gmeans(Y,Y')$, $\sigma\in C(Y,\positivedefiniteoperators{\mathcal{H}})$ and $T:\boundeds(\mathcal{H})\to\boundeds(\mathcal{H}')$ is a completely positive trace-preserving map, then $M(T(\sigma))\ge T(M(\sigma))$. To see this, consider the Stinespring dilation of $T$, and write the partial trace over the environment, followed by tensoring with the maximally mixed state, as a convex combination of unitary conjugations. On numbers ($\mathcal{H}=\complexes$) every $M\in\gmeans(Y)$ has the form $M(\sigma)=\exp\int_Y\ln\sigma\ed\gamma$ for some probability measure $\gamma$.

An example of an element of $\gmeans(\{1,2\})$ is given by 
\begin{equation}
    \sigma(1)\#\sigma(2)=\sigma(1)^{1/2}(\sigma(1)^{-1/2}\sigma(2)\sigma(1)^{-1/2})^{1/2}\sigma(1)^{1/2},
\end{equation}
the (unweighted) geometric mean of a pair of matrices, introduced in \cite{pusz1975functional} and put in a general context by Kubo and Ando \cite{kubo1980means}. Extensions to several variables have been constructed by building on the bivariate geometric mean or generalizing characterizations thereof (see e.g. \cite{moakher2005differential,petz2005means,bhatia2006riemannian,lawson2011monotonic,lim2012matrix}), and also studied from an axiomatic point of view \cite{ando2004geometric}. Examples of elements of $\gmeans\{1,2,\ldots,n\}$ are the Karcher means \cite{lawson2011monotonic,lim2012matrix}. Since our axioms in Definition~\ref{def:mean} are directly motivated by their use in constructing monotone homomorphisms, they differ from the ones considered in the literature on geometric means of matrices, in particular in their emphasis on relating the means of matrices of different sizes (by the tensor product or the direct sum). In addition, we need to consider every possible weighting of the arguments, therefore symmetry is not a relevant property in our problem.

The following proposition lists basic constructions that allows one to exhibit many elements of $\gmeans(Y)$. Geometric means that can be obtained in this way include the Ando--Li--Mathias mean \cite{ando2004geometric} and the Bini--Meini--Polini means \cite{bini2010effective}.
\begin{proposition}\label{prop:meanconstructions}
Let $Y,Y',Y''$ be nonempty compact spaces.
\begin{enumerate}
    \item\label{it:meancomposition} If $M\in\gmeans(Y,Y')$ and $N\in\gmeans(Y',Y'')$, then $N\circ M\in\gmeans(Y,Y'')$ (here the composition $N\circ M$ is understood separately for all $\mathcal{H}$).
    \item\label{it:meancontinuousmap} If $f:Y'\to Y$ is a continuous map, then $M(\sigma)=\sigma\circ f$ defines an element of $\gmeans(Y,Y')$.
    \item\label{it:meanbivariate} $M(\sigma)=\sigma'$ with
    \begin{equation}
        \sigma'(y_1,y_2,\gamma):=\sigma(y_1)\#_{\gamma}\sigma(y_2)=\sigma(y_1)^{1/2}\left(\sigma(y_1)^{-1/2}\sigma(y_2)\sigma(y_1)^{-1/2}\right)^\gamma\sigma(y_1)^{1/2}
    \end{equation}
    defines an element of $\gmeans(Y,Y\times Y\times[0,1])$.
    \item\label{it:meancompact} $\gmeans(Y)$ is compact with respect to the pointwise convergence (i.e. convergence of $i\mapsto M_i(\sigma)$ for all $\sigma$).
\end{enumerate}
\end{proposition}
\begin{proof}
\ref{it:meancomposition}: The composition is clearly additive, multiplicative, homogeneous and monotone. For concavity, we apply $N$ to the inequality $M(\lambda\sigma+(1-\lambda)\sigma')\ge\lambda M(\sigma)+(1-\lambda)M(\sigma')$, which expresses the concavity of $M$, using that $N$ is monotone:
\begin{equation}
\begin{split}
    (N\circ M)(\lambda\sigma+(1-\lambda)\sigma')
        & = N(M(\lambda\sigma+(1-\lambda)\sigma'))  \\
        & \ge N(\lambda M(\sigma)+(1-\lambda)M(\sigma'))  \\
        & \ge \lambda N(M(\sigma))+(1-\lambda)N(M(\sigma')).
\end{split}
\end{equation}

\ref{it:meancontinuousmap}: $M$ is clearly additive, multiplicative, homogeneous, monotone and affine (hence concave).

\ref{it:meanbivariate}: $\sigma'$ is clearly continuous for every continuous $\sigma$. The geometric mean is clearly additive, multiplicative and homogeneous. For concavity and monotonocity see \cite{kubo1980means} and \cite[Theorem 37.1]{simon2019loewner}.

\ref{it:meancompact}: If $\sigma\in C(Y,\positivedefiniteoperators{\mathcal{H}})$, then there exist constants $c_1,c_2>0$ such that for all $y\in Y$ the inequalities $c_1I\le\sigma(y)\le c_2I$ hold. By the direct sum and monotonicity properties, it follows that $c_1I\le M(\sigma)\le c_2I$ for every $M\in\gmeans(Y)$. The interval $[c_1I,c_2I]=\setbuild{A\in\positivedefiniteoperators{\complexes^d}}{c_1I\le A\le c_2I}$ is compact for every $d$, therefore the evaluations embed $\gmeans(Y)$ into the compact space $\prod_{d\in\naturals}\prod_{\sigma\in C(Y,\positivedefiniteoperators{\complexes^d})}[c_1(\sigma)I,c_2(\sigma)I]$. The conditions defining $\gmeans$ are closed (equalities and non-strict inequalities with respect to the semidefinite partial order), therefore the image under the embedding is closed.
\end{proof}

\begin{proposition}\label{prop:homomorphismfrommean}
Let $X,Y,X',Y'$ be nonempty compact spaces, $M\in\gmeans(Y,Y')$ and $f:X'\to X$ continuous. Then the map $(\rho,\sigma)\mapsto(\rho\circ f,M(\sigma))$ is a monotone semiring homomorphism
\end{proposition}
\begin{proof}
This map is by definition additive and multiplicative. We have yet to show monotonocity. Suppose that the completely positive trace-nonincreasing map $T$ realizes  $(\rho,\sigma)\geq (\rho',\sigma')$. Then $T\left(\rho\right)\geq \rho'$ and $T(\sigma)\leq \sigma'$. From monotonocity of $M$ in its variables under completely positive trace-nonincreasing maps:
\begin{align}
T(M(\sigma))\leq M(T(\sigma))\leq M(\sigma').
\end{align}
This yields $(\rho,M(\sigma))\geq (\rho',M(\sigma'))$ by the same map $T$.
\end{proof}

\begin{theorem}\label{thm:quantumspectrum}
Let $X,Y$ be nonempty compact spaces. For all $\alpha\ge 1$, $x\in X$ and $M\in\gmeans(Y)$ the functional
\begin{equation}\label{eq:realpointfrommean}
    f(\rho,\sigma)=\sandwichedquasientropy{\alpha}{\rho(x)}{M(\sigma)}
\end{equation}
is an element of the real part of the spectrum, and
\begin{equation}\label{eq:tropicalpointfrommean}
    f(\rho,\sigma)=\norm[\infty]{M(\sigma)^{-1/2}\rho(x)M(\sigma)^{-1/2}}
\end{equation}
is an element of the tropical part.
\end{theorem}
\begin{proof}
By Proposition~\ref{prop:homomorphismfrommean}, the map $(\rho,\sigma)\mapsto(\rho(x),M(\sigma))$ determines a monotone semiring homomorphism from $\familysemiring{X}{Y}$ to $\familysemiring{1}{1}$, where $1$ is a one-point space. On $\familysemiring{1}{1}$ the functionals $f_\alpha(\rho,\sigma)=\sandwichedquasientropy{\alpha}{\rho}{\sigma}$ are in the real spectrum and $(\rho,\sigma)\mapsto\norm[\infty]{\sigma^{-1/2}\rho\sigma^{-1/2}}$ is in the tropical spectrum, as follows from Proposition~\ref{prop:commutingsigma} (see also \cite{perry2020semiring,bunth2021asymptotic}). Therefore \eqref{eq:realpointfrommean} and \eqref{eq:tropicalpointfrommean} are compositions of monotone semiring homomorphisms, which implies that they are points in the real (respectively tropical) part of the spectrum.
\end{proof}

In the special case when $X$ is a one-point space and $Y$ has two elements, the $\alpha\to 1$ limit has recently found application in composite binary state discrimination \cite{mosonyi2020error}.

While Theorem~\ref{thm:quantumspectrum} identifies a vast collection of elements of $\testspectrum(\familysemiring{X}{Y})$, it still provides an incomplete picture of the spectrum. Our results highlight several open problems:
\begin{itemize}
    \item Is every real spectral point of the form $f(\rho,\sigma)=\sandwichedquasientropy{\alpha}{\rho(x)}{P(\sigma)}$ for some map $P:C(Y,\positivedefiniteoperators{\mathcal{H}})\to\positivedefiniteoperators{\mathcal{H}}$ (defined in a consistent way for all $\mathcal{H}$)?
    \item Assuming that a real spectral point does have the form $f(\rho,\sigma)=\sandwichedquasientropy{\alpha}{\rho(x)}{P(\sigma)}$, is $P$ necessarily an element of $\gmeans(Y)$?
    \item Classify the elements of $\gmeans(Y)$.
\end{itemize}
The analogous questions for tropical points are also interesting.

\section{Applications}\label{sec:applications}

\subsection{Equivariant relative submajorization}\label{sec:equivariant}

In this section we consider pairs of operators on a representation space of some fixed group, and a variant of relative submajorization that takes into account the group actions. Let $G$ be a topological group. Let $\pi$ and $\pi'$ be finite dimensional unitary representations of $G$ on $\mathcal{H}$ and $\mathcal{H}'$, respectively. Suppose that $(\rho_0,\sigma_0)\in\positivedefiniteoperators{\mathcal{H}}^2$ and $(\rho'_0,\sigma'_0)\in\positivedefiniteoperators{\mathcal{H}'}^2$. We say that $(\pi,\rho_0,\sigma_0)$ equivariantly relatively submajorizes $(\pi',\rho'_0,\sigma'_0)$ if there exists a completely positive trace-nonincreasing map $T:\boundeds(\mathcal{H})\to\boundeds(\mathcal{H}')$ such that
\begin{align}
T(\rho_0) & \ge \rho'_0  \\
T(\sigma_0) & \le \sigma'_0  \\
\forall g\in G\forall A\in\boundeds(\mathcal{H}): T(\pi(g)A\pi(g)^*) & = \pi'(g)T(A)\pi'(g)^*
\end{align}
On these triples the direct sum and tensor product (of representations and of operators) give binary operations that are compatible with equivariant relative submajorization.

It will be convenient to restrict to compact groups $G$, and it can be done without loss of generality for the following reason. Consider the closure $K$ of $\setbuild{(\pi(g),\pi'(g))}{g\in G}\subseteq U(\mathcal{H})\times U(\mathcal{H}')$. This is a compact group (in fact, a Lie group), the map $g\mapsto(\pi(g),\pi'(g))$ is a homomorphism and the representations $\pi,\pi'$ of $G$ extend to representations of $K$ (namely the first and second projections provide the required homomorphisms). By continuity, a map $T:\boundeds(\mathcal{H})\to\boundeds(\mathcal{H'})$ is $G$-equivariant iff it is $K$-equivariant. Therefore the condition for $(\pi,\rho_0,\sigma_0)\preorderge(\pi',\rho'_0,\sigma'_0)$ can be formulated in terms of the compact group $K$ instead of $G$. Note that in this case the compact group in general depends on the specific pair of triples to be compared (through the representations), which may not always be desirable. Alternatively, one may construct $K$ in a universal way, by taking the Bohr compactification of $G$. Recall that the Bohr compactification of a topological group $G$ is a compact Hausdorff topological group $\Bohrcpt(G)$ together with a continuous homomorphism $b:G\to\Bohrcpt(G)$ that is universal in the sense that every continuous homomorphism from $G$ into a compact group factors through $b$ in a unique way. Every topological group has an essentially unique Bohr compactification. We can apply the universal property to the homomorphisms $\pi:G\to U(\mathcal{H})$ to get a representation $\Bohrcpt(\pi):\Bohrcpt(G)\to U(\mathcal{H})$. Thus, instead of each triple $(\pi,\rho_0,\sigma_0)$ we may consider the modified triple $(\Bohrcpt(\pi),\rho_0,\sigma_0)$. For notational simplicity, from now on we will assume that $G$ itself is a compact Hausdorff group.

We now show how to map the triples $(\pi,\rho_0,\sigma_0)$ to pairs of families in such a way that the operations are preserved and equivariant relative submajorization translates to the relative submajorization of the families. In this way a triple $(\pi,\rho_0,\sigma_0)$ gives rise to the following pair of families, parametrized by $G$:
\begin{align}
    \rho(g) & = \pi(g)\rho_0\pi(g)^*  \\
    \sigma(g) & = \pi(g)\sigma_0\pi(g)^*.
\end{align}
$(\rho,\sigma)$ determines an element of $\familysemiring{G}{G}$, and this element remains the same if we replace the triple $(\pi,\rho_0,\sigma_0)$ by a unitary equivalent one. This map clearly respects the sum and product operations.

Suppose that $(\pi,\rho_0,\sigma_0)$ equivariantly relatively submajorizes $(\pi',\rho'_0,\sigma'_0)$, and let $T$ be an equivariant completely positive trace-nonincreasing map satisfying $T(\rho_0)\ge\rho'_0$ and $T(\sigma_0)\le\sigma'_0$. Consider the corresponding elements $(\sigma,\rho)$ and $(\sigma',\rho')$ of $\familysemiring{G}{G}$. Then for all $g\in G$ the inequality
\begin{equation}
    T(\rho(g))=T(\pi(g)\rho_0\pi(g)^*)=\pi'(g)T(\rho_0)\pi'(g)^*\ge\pi'(g)\rho'_0\pi'(g)^*=\rho'(g)
\end{equation}
holds and similarly $T(\sigma(g))\le\sigma'(g)$. This means that $(\rho,\sigma)\preorderge(\rho',\sigma')$ holds.

Conversely, suppose that $(\rho,\sigma)\preorderge(\rho',\sigma')$ is true in $\familysemiring{G}{G}$ for the families defined above. This means that there exists a (not necessarily equivariant) completely positive trace-nonincreasing map $T_0$ such that for all $g\in G$ the inequalities $T_0(\rho(g))\ge\rho'(g)$ and $T_0(\sigma(g))\le\sigma'(g)$ hold. We construct an equivariant map $T$ by averaging:
\begin{equation}
    T(X)=\int_G\pi'(g)^*T_0(\pi(g)X\pi(g)^*)\pi'(g)\ed\mu(g),
\end{equation}
where $\mu$ is the Haar probability measure on $G$. Then $T$ is $G$-equivariant and in addition
\begin{equation}
\begin{split}
T(\rho_0)
 & = \int_G\pi'(g)^*T_0(\pi(g)\rho_0\pi(g)^*)\pi'(g)\ed\mu(g)  \\
 & = \int_G\pi'(g)^*T_0(\rho(g))\pi'(g)\ed\mu(g)  \\
 & \ge \int_G\pi'(g)^*\rho'(g)\pi'(g)\ed\mu(g)=\rho'_0,
\end{split}
\end{equation}
and similarly $T(\sigma_0)\le\sigma'_0$.

Note that even though the map $(\pi,\rho_0,\sigma_0)\mapsto(\rho,\sigma)$ is order-preserving and order-reflecting, it is in general not injective (on equivalence classes). Now we can apply our result on general pairs of families to the question of asymptotic equivariant relative submajorization.
\begin{theorem}\label{thm:equivariantasymptotic}
Let $G$ be a topological group and consider the triples $(\pi,\rho_0,\sigma_0)$ and $(\pi',\rho'_0,\sigma'_0)$, where $\pi$ is a unitary representation of $G$ on $\mathcal{H}$, $\rho_0,\sigma_0$ are positive definite operators on $\mathcal{H}$ and similarly for $\pi',\rho'_0,\sigma'_0$ on $\mathcal{H}'$. The following are equivalent:
\begin{enumerate}
    \item there exists a sequence of $G$-equivariant completely positive trace-nonincreasing maps $T_n:\boundeds(\mathcal{H})\to\boundeds(\mathcal{H}')$ such that for all $x\in X$ the inequalities
    \begin{align}
        T_n(\rho_0^{\otimes n}) & \ge 2^{-o(n)}{\rho'_0}^{\otimes n}  \\
        T_n(\sigma_0^{\otimes n}) & \le {\sigma'_0}^{\otimes n}
    \end{align}
    hold, with the $o(n)$ uniform in $x$,
    \item $f((\pi(g)\rho_0\pi(g)^*)_{g\in G},(\pi(g)\sigma_0\pi(g)^*)_{g\in G})\ge f((\pi(g)\rho'_0\pi(g)^*)_{g\in G},(\pi(g)\sigma'_0\pi(g)^*)_{g\in G})$ for all $f\in\testspectrum(\familysemiring{G}{G})$.
\end{enumerate}
\end{theorem}
We note that in general many elements of $\testspectrum(\familysemiring{G}{G})$ collapse to the same function when restricted to pairs of the form $((\pi(g)\rho_0\pi(g)^*)_{g\in G},(\pi(g)\sigma_0\pi(g)^*)_{g\in G})$. The reason is that left translations of $G$ give rise to automorphisms of $\familysemiring{G}{G}$ of the form $(\rho,\sigma)\mapsto(\rho\circ L_h,\sigma\circ L_h)$ (where $h\in G$ and $L_h:G\to G$ is the map $L_h(g)=hg$), which in turn induce nontrivial automorphisms of $\testspectrum(\familysemiring{G}{G})$, while the equivalence class of $((\pi(g)\rho_0\pi(g)^*)_{g\in G},(\pi(g)\sigma_0\pi(g)^*)_{g\in G})$ remains unchanged by these transformations. This can be seen explicitly on the subsemiring of pairs with commuting $\sigma$, where the precise form of spectral points is known: if $\rho(g)=\pi(g)\rho_0\pi(g)^*$ and $\sigma(g)=\pi(g)\sigma_0\pi(g)^*$ such that $\sigma(g)\sigma(e)=\sigma(e)\sigma(g)$ for all $g\in G$, then
\begin{equation}
    \begin{split}
        f_{\alpha,x,\gamma}(\rho\circ L_h,\sigma\circ L_h)
         & = \sandwichedquasientropy{\alpha}{\rho(hx)}{\exp\int_G\ln\sigma(hg)\ed\gamma(g)}  \\
         & = \sandwichedquasientropy{\alpha}{\pi(h)\rho(x)\pi(h)^*}{\pi(h)\exp\int_G\ln\sigma(g)\ed\gamma(g)\pi(h)^*}  \\
         & = \sandwichedquasientropy{\alpha}{\rho(x)}{\exp\int_G\ln\sigma(g)\ed\gamma(g)}  \\
         & = f_{\alpha,x,\gamma}(\rho,\sigma).
    \end{split}
\end{equation}
The first line of this calculation also shows that
\begin{equation}
    f_{\alpha,x,\gamma}(\rho\circ L_h,\sigma\circ L_h)=f_{\alpha,hx,(L_h)_*(\gamma)}(\rho,\sigma).
\end{equation}
In particular, $f_{\alpha,h,\gamma}$ and $f_{\alpha,e,(L_{h^{-1}})_*(\gamma)}$ coincide on these elements.
\begin{corollary}\label{cor:equivariantasymptoticcommuting}
Under the conditions of Theorem~\ref{thm:equivariantasymptotic}, suppose that $[\sigma(e),\sigma(g)]=0$ for all $g\in G$. Then $(\pi,\rho_0,\sigma_0)\asymptoticge(\pi',\rho'_0,\sigma'_0)$ (in the sense of asymptotic equivariant relative submajorization) iff for all $\alpha\ge 0$ and Radon probability measure $\gamma$ on $G$ the inequality
\begin{equation}
    \sandwicheddivergence{\alpha}{\rho_0}{\exp\int_G\ln\pi(g)\sigma_0\pi(g)^*\ed\gamma(g)}\ge\sandwicheddivergence{\alpha}{\rho'_0}{\exp\int_G\ln\pi(g)\sigma'_0\pi(g)^*\ed\gamma(g)}
\end{equation}
holds.
\end{corollary}

\subsubsection{Asymptotic transformations by thermal processes}
\label{sec:thermal}
Thermal operations are central to the resource theoretic approach to quantum thermodynamics. This is the class of quantum channels that can be obtained by preparing Gibbs states at a fixed temperature $T$, performing energy-preserving and tracing out subsystems \cite{janzing2000thermodynamic,brandao2013resource,horodecki2013fundamental}. This characterization does not suggest a simple way to decide whether a given channel is a thermal operation or whether a transformation between given states is feasible by a thermal operation, which motivates the study of channels and transformations admitting a simpler description at the cost of satisfying only some of the constraints governing thermal operations. In the absence of coherence between energy eigenspaces, Gibbs-preserving maps provide an especially useful relaxation, which turns out to allow the same transitions as thermal operations. This is no longer true if coherence is present \cite{faist2015gibbs}, and in addition to being Gibbs-preserving, the condition of time-translation symmetry has been identified as another key property of thermal operations \cite{lostaglio2015description}. Adding this requirement leads to the notion of thermal processes \cite{gour2018quantum}. Transformations by such processes are an instance of equivariant relative majorization: the group is that of time-translations, isomorphic to $\reals$, and to a system with Hilbert space $\mathcal{H}$, Hamiltonian $H\in\boundeds(\mathcal{H})$ and state $\rho$ we associate the triple $(\pi_H,\rho,e^{-\beta H})$, where $\pi_H:\reals\to U(\mathcal{H})$ is the representation $t\mapsto e^{-itH}$ and $\beta$ is the inverse temperature. If we relax these transformations to equivariant relative \emph{sub}majorization and consider the asymptotic limit, then the problem reduces to evaluating the elements of $\testspectrum(\familysemiring{\Bohrcpt(\reals)}{\Bohrcpt(\reals)})$ on the pairs corresponding to the initial and target triple. While this might look difficult at first sight, the task is greatly simplified by the fact that $\pi_H(t)e^{-\beta H}\pi_H(t)^*$ is constant. This means that the second argument of the spectral points is commuting (so we may use Corollary~\ref{cor:equivariantasymptoticcommuting}), and since we need to integrate a constant function, the expressions do not depend on the probability measure $\gamma$. This means that the spectral points essentially reduce to $\sandwichedquasientropy{\alpha}{\rho}{e^{-\beta H}}$ as in \cite{perry2020semiring}, implying that in this limit Gibbs-preserving maps are no more powerful than thermal processes.

\subsubsection{Hypothesis testing with group symmetry}
The task in asymptotic binary state discrimination is to decide, based on measurements on many copies, if the state of a system is $\rho_0$ or $\sigma_0$, with the promise that it is one of the two. A type I error occurs if $\rho_0$ is accepted but the state in fact was $\sigma_0$, the opposite case is called the type II error. In the asymptotic setting, one is interested in the restrictions on the limiting behaviours of the two kinds of errors. In the group-symmetric variant of this problem, the measurements are restricted to be invariant with respect to a group representation $\pi:G\to U(\mathcal{H})$ \cite{hiai2009quantum}. An invariant measurement on $\mathcal{H}^{\otimes n}$ can be identified with an equivariant (completely) positive trace-nonincreasing map $T:\boundeds(\mathcal{H}^{\otimes n})\to\boundeds(\complexes)$, where $\complexes$ carries the trivial representation $1$. Following the ideas of \cite{perry2020semiring}, the strong converse error exponent can be characterized in terms of the asymptotic preorder: there exists a sequence of invariant measurements such that the type I error behaves like $1-2^{-Rn+o(n)}$ and the type II error is at most $2^{-rn+o(n)}$ iff $(\pi,\rho_0,\sigma_0)\asymptoticge(1,2^{-R},2^{-r})$. Theorem~\ref{thm:equivariantasymptotic} gives the necessary and sufficient condition
\begin{equation}
    \forall f\in\testspectrum(\familysemiring{G}{G}):f((\pi(g)\rho_0\pi(g)^*)_{g\in G},(\pi(g)\sigma_0\pi(g)^*)_{g\in G})\ge f(2^{-R},2^{-r})
\end{equation}
in the general case, while Corollary~\ref{cor:equivariantasymptoticcommuting} gives
\begin{equation}\label{eq:invariantstrongconverse}
    R\ge\sup_{\alpha>1}\max_\gamma\frac{\alpha-1}{\alpha}\left[r-\sandwicheddivergence{\alpha}{\rho_0}{\exp\int_G\ln\pi(g)\sigma_0\pi(g)^*\ed\gamma(g)}\right]
\end{equation}
when the orbit of $\sigma_0$ consists of operators commuting with $\sigma_0$.

\subsubsection{Reference frames in hypothesis testing}
When the dynamics of a system is constrained by symmetries, an additional supply of asymmetric states (imperfect reference frames) becomes a resource, which allows to partially overcome the limitations of symmetric evolutions \cite{bartlett2007reference}. Suppose that $\pi_{\textnormal{ref}}:G\to U(\mathcal{K})$ is a representation and $\Omega\in\states(\mathcal{K})$ is a state with full support and trivial stabilizer. In the setting of group-symmetric hypothesis testing as modeled above in terms of equivariant relative submajorization, the reference frame corresponds to the triple $(\pi_{\textnormal{ref}},\Omega,\Omega)$. Testing $(\pi,\rho,\sigma)$ aided by the reference frame corresponds to comparing $(\pi\otimes\pi_{\textnormal{ref}},\rho\otimes\Omega,\sigma\otimes\Omega)$ with a one-dimensional triple with trivial representation. In an asymptotic setting, more generally, we may use $\kappa$ copies of the reference frame per sample of the state to be discriminated. In this case the exponent pair $(R,r)$ is achievable iff
\begin{multline}
    \forall f\in\testspectrum(\familysemiring{G}{G}):f((\pi(g)\rho_0\pi(g)^*)_{g\in G},(\pi(g)\sigma_0\pi(g)^*)_{g\in G})  \\  \cdot  f((\pi_{\textnormal{ref}}(g)\Omega\pi_{\textnormal{ref}}(g)^*)_{g\in G},(\pi_{\textnormal{ref}}(g)\Omega\pi_{\textnormal{ref}}(g)^*)_{g\in G})^\kappa\ge f(2^{-R},2^{-r}).
\end{multline}
When the orbits of $\sigma_0$ and $\Omega$ consist of commuting operators, we can use Corollary~\ref{cor:equivariantasymptoticcommuting} to obtain an explicit form of the smallest type I strong converse exponent $R^*$ for a given decay rate $r$ of the type II error:
\begin{multline}
    R^*(r,\kappa)=\sup_{\alpha>1}\max_\gamma\frac{\alpha-1}{\alpha}\bigg[r-\sandwicheddivergence{\alpha}{\rho_0}{\exp\int_G\ln\pi(g)\sigma_0\pi(g)^*\ed\gamma(g)}  \\  -\kappa\sandwicheddivergence{\alpha}{\Omega}{\exp\int_G\ln\pi_{\textnormal{ref}}(g)\Omega\pi_{\textnormal{ref}}(g)^*\ed\gamma(g)}\bigg].
\end{multline}
As $\kappa\to\infty$ (ie. in the limit of unlimited supply of the reference frame), the supremum is achieved for the $\gamma$ that is concentrated on the identity element of the $G$ (by our assumption of $\Omega$ having trivial stabilizer), since this is the only point where the last term vanishes. This means that we recover the unrestricted strong converse exponent \cite{mosonyi2015quantum}, which is potentially much smaller than the group-symmetric one. In an extreme example, $\rho_0$ and $\sigma_0$ might be in the same $G$-orbit, in which case $R^*(r,0)=r$, ie. a group-symmetric measurement cannot offer any advantage over guessing.

\subsection{Approximate joint transformations}\label{sec:approximate}

In this section we specialize our results and derive a characterization of approximate joint transformations with respect to the max-divergence, in the asymptotic limit. Recall that the max-divergence between a pair of states $\rho,\sigma$ is $\maxrelativeentropy{\rho}{\sigma}=\log\norm[\infty]{\sigma^{-1/2}\rho\sigma^{-1/2}}=\min\setbuild{\lambda\in\reals}{2^\lambda\sigma\ge\rho}$. We will use the max-divergence as a measure of dissimilarity between states. It vanishes iff the two states are equal, but for subnormalized states this is no longer true, and it is not symmetric. However, the closely related quantity $\maxdistance(\rho,\sigma):=\max\{\maxrelativeentropy{\rho}{\sigma},\maxrelativeentropy{\sigma}{\rho}\}$ is a metric on the set of positive definite operators (the Thompson metric \cite{thompson1963certain} associated with the semidefinite cone). This metric is unbounded even on a fixed Hilbert space and satisfies $\maxdistance(\rho^{\otimes n},\sigma^{\otimes n})=n\maxdistance(\rho,\sigma)$. The notion of approximate transformations that we consider will be that the distance increases sublinearly as the number of copies grow.

This problem fits in our framework in the following way: we specialize to pairs of families on the same space $X=Y$, and compare elements of the form $(\rho,\rho)$, where $\rho\in C(X,\positivedefiniteoperators{\mathcal{H}})$. Relative submajorization between such pairs $(\rho,\rho)$ and $(\rho',\rho')$ means the existence of a completely positive trace-nonincreasing map $T$ such that $T(\rho)\ge\rho'$ and $T(\rho)\le\rho'$, ie. $T(\rho)=\rho'$.

Recall that the asymptotic preorder $\asymptoticge$ is defined by allowing a sublinear number of copies of the power universal element $u=(2\cdot 1_X,1_X)$. Therefore $(\rho,\rho)\asymptoticge(\rho',\rho')$ iff there is a sequence of completely positive trace-nonincreasing maps $T_n:\boundeds(\mathcal{H}^{\otimes n})\to\boundeds({\mathcal{H}'}^{\otimes n})$ such that for all $n$
\begin{equation}
    2^{-o(n)}{\rho'}^{\otimes n}\le T_n(\rho^{\otimes n})\le{\rho'}^{\otimes n},
\end{equation}
i.e. for all $x\in X$ we have
\begin{equation}
    \lim_{n\to\infty}\frac{1}{n}\maxdistance(T_n(\rho(x)^{\otimes n}),\rho'(x)^{\otimes n})=0,
\end{equation}
where the limit is uniform in $x$. The condition for this is that for all $f\in\testspectrum(\familysemiring{X}{Y})$ the inequality $f(\rho,\rho)\ge f(\rho',\rho')$ holds.

Specializing to classical families and using the explicit form of the 1-test spectrum of the classical semiring, we have the following characterization of asymptotic joint transformations in the above sense.
\begin{theorem}
Let $p:X\to\distributions([d])$, $p':X\to\distributions([d'])$ where $X$ is a compact Hausdorff space and $d,d'\in\positiveintegers$ are finite sets. The following are equivalent:
\begin{enumerate}
    \item there exists a sequence of substochastic maps $T_n$ from $(\reals^d)^{\otimes n}$ to $(\reals^{d'})^{\otimes n}$ such that for all $x\in X$
    \begin{equation}
        \lim_{n\to\infty}\frac{1}{n}\maxdistance(T_n(p(x)^{\otimes n}),p'(x)^{\otimes n})=0,
    \end{equation}
    uniformly in $x$;
    \item for all $x\in X$, $\alpha\ge 1$ and probability measure $\gamma$ on $X$ the inequality
    \begin{equation}
        \sum_{i=1}^dp_i(x)^\alpha\exp(1-\alpha)\int_X \ln p_i\ed\gamma\ge\sum_{i=1}^{d'}p'_i(x)^\alpha\exp(1-\alpha)\int_X \ln p'_i\ed\gamma
    \end{equation}
    holds.
\end{enumerate}
\end{theorem}

We note that the ideas in this section can be combined with our considerations on equivariant transformations. If $X_0$ is a compact space, $G$ a compact group, $\pi:G\to U(\mathcal{H})$, $\pi':G\to U(\mathcal{H}')$ are unitary representations, $\rho_0\in C(X_0,\positivedefiniteoperators{\mathcal{H}})$ and $\rho'_0\in C(X_0,\positivedefiniteoperators{\mathcal{H}'})$, then we may ask whether an equivariant completely positive trace-nonincreasing map $T$ exists such that $T(\rho_0(x))=\rho'_0(x)$ for all $x\in X_0$. This can be encoded in the relative submajorization of families as follows. Let $X=X_0\times G$ and consider $\rho\in C(X,\positivedefiniteoperators{\mathcal{H}})$ defined as $\rho(x,g)=\pi(g)\rho_0(x)\pi(g)^*$, and similarly $\rho'$, which determine the elements $(\rho,\rho)$ and $(\rho',\rho')$ in $\familysemiring{X}{X}$. Then the existence of a suitable $T$ is equivalent to $(\rho,\rho)\preorderge(\rho',\rho')$. In the setting of Section~\ref{sec:thermal}, the spectrum gives rise to many ``second laws'' of thermodynamics in the sense of \cite{brandao2015second}, for example the value of
\begin{equation}\label{eq:thermalmonotone}
    \sandwicheddivergence{\alpha}{\rho}{(e^{-itH}\rho e^{itH})\#_\gamma e^{-\beta H}}
\end{equation}
cannot increase under a thermal process, for every $\alpha\ge 1$, $\gamma\in[0,1]$ and $t\in\reals$. Here the second argument may be replaced with any weighted geometric mean of the Gibbs state and arbitrary time-translated versions of $\rho$, and in addition the first argument may be replaced with the Gibbs state. To ensure that the quantity is finite, one generally needs to restrict to full-rank states $\rho$.

As a concrete example, consider a transition studied in \cite{faist2015gibbs}, perturbed slightly to get full-rank states. With the Hamiltonian $H=\ketbra{1}{1}$ on $\complexes^2$, the transition $\ketbra{1}{1}\to\ketbra{+}{+}$ was shown to be possible by Gibbs-preserving maps, but not possible with thermal processes. Let $\tau=e^{-\beta H}/\Tr e^{-\beta H}$ be the Gibbs state at temperature $\beta^{-1}$. Then the transition $(1-\epsilon)\ketbra{1}{1}+\epsilon\tau\to(1-\epsilon)\ketbra{+}{+}+\epsilon\tau$ is still possible by a Gibbs-preserving map. However, \eqref{eq:thermalmonotone} with $\alpha=2$, $\gamma=1/2$ and $t=\pi$ evaluates to $\frac{1}{2}\log(1+e^\beta)+O(\epsilon)$ on the initial state while it diverges logarithmically on the target state as $\epsilon\to 0$. This implies that, for sufficiently small $\epsilon$, the transition is not possible under a thermal process, even in the presence of a catalyst or assuming multi-copy transformations. In fact, a numerical comparison suggests that this holds for all $\epsilon\in(0,1)$.

\subsection{A two-parameter family of quantum R\'enyi divergences}\label{sec:divergences}

A defining property of the monotone quantities in $\testspectrum(\familysemiring{X}{Y})$ is that they are increasing in the first argument and decreasing in the second one. From the point of view of relative majorization, this is a severe and unnecessary restriction, and it is reasonable to expect that by dropping this requirement one gets more constraints on joint transformations. We now point out that it is possible to derive some of these additional constraints by specialization, thanks to the possibility of relative submajorization to express relative majorization as a special case. We also used this in Section~\ref{sec:approximate} for classical families, but now with a different viewpoint, we consider quantum pairs instead, and introduce a two-parameter family of monotone quantum R\'enyi divergences. We note that $\alpha$-$z$-divergences, another two-parameter quantum extension of the R\'enyi divergences introduced in \cite{audenaert2015alpha}, do not seem have any obvious relation to ours.

To this end, we let $X=Y=\{1,2\}$ and consider pairs of identical families (where the family means pair in this case). We change the notation to reflect the different point of view: in both families the first element will be denoted by $\rho$ and the second one $\sigma$ (so the pair of families may be written as $((\rho,\sigma),(\rho,\sigma))$). Recall that in the quantum case we only have a partial understanding of the spectrum, summarized in Theorem~\ref{thm:quantumspectrum} and Proposition~\ref{prop:meanconstructions}. It is known that if we start with only a pair of positive operators, then any iterated weighted geometric mean is equal to a single weighted geometric mean, so $M(\rho,\sigma)=\sigma\#_\gamma\rho=\rho\#_{1-\gamma}\sigma$ for some $\gamma\in[0,1]$ (there may be other functions satisfying the axioms in Definition~\ref{def:mean}, but these are the ones that one can construct via Proposition~\ref{prop:meanconstructions}). This implies that $\sandwichedquasientropy{\alpha_0}{\rho}{\sigma\#_\gamma\rho}$ and $\sandwichedquasientropy{\alpha_0}{\sigma}{\sigma\#_\gamma\rho}$ are monotone decreasing under joint application of completely positive and trace-\emph{preserving} maps, for all $\alpha_0\ge 1$ and $\gamma\in[0,1]$. On commuting pairs the first one reduces to the R\'enyi quasientropy of order $\alpha_0+\gamma(1-\alpha_0)$, which suggests introducing the following sandwiched-geometric R\'enyi divergences as a generalization of the sandwiched R\'enyi divergence:
\begin{equation}\label{eq:sandwichedmeandivergence}
    \begin{split}
        \sandwichedmeandivergence{\alpha,\gamma}{\rho}{\sigma}
         & := \frac{1}{1-\gamma}\sandwicheddivergence{\frac{\alpha-\gamma}{1-\gamma}}{\rho}{\sigma\#_\gamma\rho}  \\
         & = \frac{1}{\alpha-1}\log\Tr\left(\sqrt{\rho}\left(\sqrt{\sigma}\left(\sigma^{-1/2}\rho\sigma^{-1/2}\right)^\gamma\sqrt{\sigma}\right)^{\frac{1-\alpha}{\alpha-\gamma}}\sqrt{\rho}\right)^{\frac{\alpha-\gamma}{1-\gamma}}
    \end{split}
\end{equation}
By construction, for all $\alpha>1$ and $\gamma\in[0,1)$ this quantity is additive under the tensor product, satisfies the data processing inequality, and is a decreasing function of the second argument (but it is \emph{not} increasing in the first argument), and reduces to the R\'enyi divergence of order $\alpha$ on commuting arguments. When $\gamma=0$, \eqref{eq:sandwichedmeandivergence} agrees with the minimal R\'enyi divergence, whereas we do not know what the limit $\gamma\to 1$ is. We leave the detailed study of these divergences for future work.

\section*{Acknowledgement}

This work was supported by the \'UNKP-20-5 New National Excellence Program of the Ministry for Innovation and Technology and the J\'anos Bolyai Research Scholarship of the Hungarian Academy of Sciences. We acknowledge support from the Hungarian National Research, Development and Innovation Office (NKFIH) within the Quantum Technology National Excellence Program (Project Nr.~2017-1.2.1-NKP-2017-00001) and via the research grants K124152, KH129601.

\newcommand{\etalchar}[1]{$^{#1}$}
\end{document}